\documentclass[11pt]{article}
\usepackage{graphicx} 
\usepackage{amsmath, amsfonts, amssymb, amsthm,xcolor}
\usepackage{fullpage}
\usepackage[colorlinks=true,citecolor=blue,linkcolor=blue]{hyperref}
\usepackage{thm-restate}
\hypersetup{
    colorlinks,
    linkcolor={red!50!black},
    citecolor={blue!50!black},
    urlcolor={blue!80!black}
}
\usepackage[nameinlink, noabbrev, capitalize]{cleveref}
\usepackage[
backend=biber,
style=alphabetic,
sorting=nyt,
maxnames=7,
maxalphanames=7,
backref=true,
url=false
]{biblatex}
\usepackage{float}
\usepackage{footnote}
\makesavenoteenv{tabular}

\newcommand{\zo}{\{0,1\}}

\newcommand{\eps}{\varepsilon}
\newcommand{\abs}[1]{\left\lvert #1 \right\rvert}

\DeclareMathOperator{\supp}{Supp}

\newtheorem{counter}{Counter}[section]
\newtheorem{theorem}[counter]{Theorem}

\newtheorem{lemma}[counter]{Lemma}

\newtheorem{claim}[counter]{Claim}

\newtheorem{definition}[counter]{Definition}

\newcommand{\N}{\mathbb{N}}

\newcommand{\F}{\mathbb{F}}

\newcommand{\E}{\mathbb{E}}
\newcommand{\poly}{\operatorname{poly}}

\newcommand{\U}{\mathbf{U}}

\newcommand{\X}{\mathbf{X}}

\newcommand{\Y}{\mathbf{Y}}

\newcommand{\cF}{\mathcal{F}}

\newcommand{\Ext}{\mathsf{Ext}}
\newcommand{\Disp}{\mathsf{Disp}}
\newcommand{\Cond}{\mathsf{Cond}}
\newcommand{\LExt}{\mathsf{LExt}}
\newcommand{\sExt}{\mathsf{sExt}}
\newcommand{\pre}{\mathsf{pre}}
\newcommand{\out}{\mathsf{out}}
\newcommand{\Nbr}{\mathsf{Nbr}}
\newcommand{\Bad}{\mathsf{Bad}}
\newcommand{\Far}{\mathsf{Far}}
\newcommand{\minH}{H_\infty}
\renewcommand{\P}{\mathsf{P}}
\newcommand{\NP}{\mathsf{NP}}
\newcommand{\BPP}{\mathsf{BPP}}
\newcommand{\OR}{\mathsf{OR}}
\newcommand{\AND}{\mathsf{AND}}

\DeclareMathOperator*{\argmax}{\arg\!\max}

\title{Sumset Structure in Local Computation}
 
 \author{
	Alexander Golovnev%
	\thanks{Cornell University. This research is supported by the National Science Foundation CAREER award (grant CCF-2338730). Email: \texttt{alexgolovnev@gmail.com}.}
	\and
	Mohit Gurumukhani%
    \thanks{Cornell University. This research is supported by a Sloan Research Fellowship, the National Science Foundation CAREER Award (grant CCF-2045576), and the National Science Foundation Award CCF-2514586. Email: \texttt{mgurumuk@cs.cornell.edu}.} 
}
\date{}
\begin{document}

\maketitle

\begin{abstract}
We introduce and study a new model of decision trees that lies at the frontier of provable circuit lower bounds.

We study \emph{$\ell$-local decision trees}, in which each internal node queries an $\ell$-local function of the input. We show that sufficiently strong lower bounds for $\ell$-local decision trees would imply several breakthrough circuit lower bounds, including super-linear size lower bounds for log-depth circuits and improved bounds for unrestricted-depth circuits. Previously, such consequences were known to follow from strong lower bounds for depth-$3$ circuits, which has been the main prior approach in attempting to prove these results. Since local decision trees are strictly weaker than depth-$3$ circuits, this provides a formally easier route to the same circuit lower-bound consequences.

We prove essentially optimal lower bounds for a weaker variant that we call \emph{oblivious $\ell$-local decision trees}, where all nodes at the same depth query the same function. 
Our lower bounds follow from a new technique that uncovers sumset structure in local maps, and our hard functions are sumset dispersers, sumset condensers, and directional affine dispersers. Along the way, we give an explicit construction of a sumset condenser with small entropy loss.

As an additional contribution, we initiate the study of \emph{degree-$2$ decision trees}, in which each query is a quadratic polynomial of the input. We show that proving lower bounds in this model is a natural stepping stone toward constructing dispersers for degree-$2$ variety sources, which imply improved circuit lower bounds for unrestricted depth circuits. We prove nearly maximal lower bounds for the oblivious variant of the model.
\end{abstract}

\pagenumbering{roman}
\thispagestyle{empty}
\newpage
\setcounter{tocdepth}{2}
\tableofcontents
\newpage
\pagenumbering{arabic}

\section{Introduction}
The central goal of circuit complexity is to prove lower bounds on the complexity of an explicit function in various models of computation. A counting argument shows that almost all functions require very large complexity in most computational models. Yet it is notoriously hard to find a polynomial-time computable function which exhibits a non-trivially large complexity. The lack of non-trivial lower bounds is used as an explanation for why progress has been stalled in other areas of complexity, such as proving $\BPP = \P$ or proving $\P \ne \NP$. 

Directly working with various circuit models in the hope of proving lower bounds is often quite difficult since it is hard to analyze them. Therefore, one often reduces the task of proving lower bounds against various circuit models to another task such as: 1) proving lower bounds for other computational models which are seemingly easier to analyze such as depth-3 circuits or the number-on-forehead model; 2) constructing explicit pseudorandom objects such as rigid matrices or dispersers for variety sources. Despite intense efforts to: improve these reductions, prove better lower bounds for seemingly simpler models, and to construct pseudorandom objects, there has been very little progress in obtaining these lower bounds. 

Motivated by this, we initiate the systematic study of \emph{local decision trees}, a natural model of computation where proving strong lower bounds yields many strong circuit lower bounds.\footnote{This model was first studied by Chistopolskaya and Podolskii and they focused on $2$-local decision trees, and their connections to communication complexity \cite{CP22}.} 
Our model has rich connections to many other models of computation and we will show that proving lower bounds against local decision trees is an easier task: that lower bounds against simple circuit models as well as explicit constructions of pseudorandom objects will yield lower bounds against these trees. We will then prove nearly maximal lower bounds against a slightly weakened version of local decision trees by finding additive structure in local computation. The hard functions for which our lower bound will hold will be sumset dispersers and condensers, which either we already know explicitly how to construct or we construct in this work for our task. 

Let us first define these trees. An \emph{$\ell$-local decision tree} over $n$ variables is a binary tree where each internal node queries an $\ell$-local function of the input and based on that branches left or right, until it reaches a leaf which is labeled with the outcome $0$ or $1$. We will prove strong lower bounds against ``oblivious'' versions of these decision trees where all nodes at the same level compute the exact same function. An \emph{oblivious $\ell$-local decision tree} $T$ of depth $d$ can also be characterized by an $\ell$-local map $L: \zo^n \to \zo^d$ and a function $h: \zo^d \to \zo$ of arbitrary complexity such that $T(x) = h(L(x))$. We also naturally extend these definitions to functions that output $m$ bits by letting each leaf be labeled by an outcome from $\zo^m$. 

We additionally study degree-$d$ decision trees---decision trees where at each node, the function queried is a degree-$d$ multilinear polynomial (over $\F_2$) of the input bits. We prove maximal lower bounds against oblivious degree-$2$ decision trees using directional affine dispersers. We show that proving these lower bounds and similar maximal lower bounds against (non-oblivious) degree-$2$ decision trees are necessary in order to construct dispersers for variety sources and obtain improved circuit lower bounds.

\subsection{Prior and Related work}

We mention a related line of work that seeks to find algebraic structure in local computation: This arises in extractor theory, where the goal is to transform entropic distributions generated by local computations into uniform distribution \cite{DW12, V14, OCGLR22}.

We now summarize several prior lower bounds for local decision trees that are immediate from prior work, or that have previously appeared in the literature.

We begin with a simple observation that for any function that depends on all $n$ variables, the depth of an $\ell$-local decision tree computing it is at least $\frac{n}{\ell}$.
Therefore, to measure the non-triviality of a function, we introduce the following quantity, defined for both oblivious and non-oblivious local decision trees.
\begin{definition}[Depth overhead]
For a function $f: \zo^n \to \zo^m$, we define its (oblivious) \emph{depth overhead} as
\[
    \max_{1\le \ell\le n}\; \frac{d_{\ell}(f)}{(n / \ell)} \;,
\]
where $d_{\ell}(f)$ is the minimum depth of an (oblivious) $\ell$-local decision tree computing it. 
\end{definition}
This quantity measures the multiplicative overhead over the trivial bound of $n / \ell$ and we will use it as a benchmark to compare and contextualize results across different values of $\ell$.\footnote{Lecomte et al. also introduced this quantity as composition overhead, which is equivalent to depth overhead for oblivious local decision trees \cite{LRT22}.}

\paragraph{Lower bounds against (non-oblivious) $\ell$-local decision trees.}
We can obtain lower bounds against $\ell$-local decision trees using the connection to depth-3 circuits mentioned previously.
The best lower bound against $\Sigma_3^{\ell}$ against an explicit function is of size $2^{c n / \ell}$ where $c \simeq 1.282$~\cite{PPSZ05} and only for $\ell = 2$, we have $2^{n - o(n)}$ lower bounds~\cite{PSZ00}. The explicit functions one can use in both cases are indicators of good codes. Also in the last case, the hard explicit function can be an affine disperser for min-entropy $o(n)$.
These lower bounds directly yield depth lower bounds of size $c n / \ell$ and $n - o(n)$ for $\ell$-local and $2$-local decision trees, respectively. 

Chistopolskaya and Podolskii studied the model of $2$-local decision trees and showed that any $2$-local decision tree computing Majority must have depth at least $n - o(n)$ \cite{CP22}.

Hence, the best depth overhead for $\ell$-local decision trees is $O(1)$.

\paragraph{Lower bounds against oblivious $\ell$-local decision trees.}

Oblivious $\ell$-local decision trees were studied (under different terminology) in the works of Hrube\v{s} and Rao, Lecomte et al., as well as Bessonov and Podolskii \cite{HR15, LRT22, BP26}. Hrube\v{s} and Rao showed that there exists an explicit function such that any oblivious $(1-\gamma) n$-local decision tree computing it, requires depth at least $n^{c \gamma}$ where $\gamma \in (0, 1)$ is an arbitrary constant and $c \simeq 0.15$ is a universal constant \cite{HR15}. Hence, this function has depth overhead of about $n^{0.15}$. Their proof strategy was inspired by Nechiporuk's lower bound on formula size \cite{N66}.

Lecomte et al. showed that any oblivious $\ell$-local decision tree computing the Majority function requires depth at least $\frac{n}{\ell}\cdot \log(\ell)$ (provided $\ell \le n^{0.99}$), obtaining depth overhead of $\log(n)$ \cite{LRT22}. They used information-theoretic arguments to obtain this lower bound. Using the connection to oblivious branching programs, they recovered the classical lower bound \cite{AM86, BPRS90} that any constant-width oblivious branching program computing Majority requires length $\ge \Omega(n \log(n))$ using vastly different techniques. 

Bessonov and Podolskii~\cite{BP26} proved that a special case of oblivious $3$-local decision trees computing Majority also require depth $n-o(n)$.\footnote{In this special case, the functions computed in the nodes of the oblivious decision tree are depth-$2$ decision trees themselves.}

\subsection{Our results}

Our main conceptual contribution is that the task of proving lower bounds for many varied circuit models can be reduced to proving lower bounds against local decision trees. And for many such models, it is necessary to do so.
We provide these connections in detail in \cref{subsec: Connections to other circuit models}.
We view our strong lower bounds against oblivious local decision trees as progress towards the task of proving such strong lower bounds against local decision trees.

We obtain lower bounds in 3 different settings: 1) Against oblivious local decision trees outputting 1 bit. 2) Against oblivious local decision trees outputting many bits. 3) Against oblivious degree-2 decision trees outputting 1 bit. The hard functions for our lower bounds in these three settings are: sumset dispersers, sumset condensers, and directional affine dispersers, respectively. 
As far as we are aware, this is the first instance of a lower bound that uses the full powers of these pseudorandom objects to make progress towards frontier circuit lower bounds.\footnote{Directional affine dispersers and sumset dispersers have been used to obtain lower bounds against linear branching programs \cite{GPT22, CL23, LZ24b, LZ24}.}
Our lower bounds will be obtained by finding sumset structure in local computation. 
Such a proof technique has not appeared in previous works and we hope it leads to interesting lower bounds against (non-oblivious) local decision trees and further circuit lower bounds.

We now present our results in each of the three settings below.

\paragraph{Single-output oblivious local decision trees.}

Our main result in this setting is as follows.
\begin{restatable}{maintheorem}{MainThmSingleOutput}\label{thm:MainSingle}
There exists a polynomial-time computable function $\Disp: \zo^n \to \zo$ and a constant $C>0$ such that any oblivious $\ell$-local decision tree computing $\Disp$ has depth at least
\[
 \frac{n^2}{n+C\ell^2\log(n)}\;.
\]
In particular, for $\ell=o(\sqrt{n/\log(n)})$, any oblivious $\ell$-local decision tree computing $\Disp$ has depth $n-o(n)$.
\end{restatable}
Thus, our depth overhead here is $\sqrt{n / \log(n)}$, much larger than the depth overhead obtained in  previous works by Hrube\v{s} and Rao as well as Lecomte et al. 
Our explicit hard function $\Disp$ is a sumset disperser and has the property that for any $A, B\subseteq \zo^n$ with $\abs{A}, \abs{B}\ge \poly(n)$, it holds that $\Disp(A + B)$ is non-constant. Such a function $\Disp$ (with even stronger guarantees) was recently constructed by Li \cite{L23}. These dispersers strictly generalize affine dispersers which were used as the hard function in proving unrestricted-depth circuit lower bounds \cite{DK11,FGHK16,LT22}. 
As far as we are aware, this is the first instance of a lower bound that uses the full power of a sumset disperser.

\paragraph{Multi-output oblivious local decision trees.}

We first mention that for functions that output $m$ bits and whose image is $\zo^m$, we trivially obtain the lower bound of $m$ on the depth of a local decision tree computing it. Our result will have the property that $d = \omega(m)$ and hence it will be non-trivial. 

We now state our nearly maximal lower bounds in this setting.
\begin{restatable}{maintheorem}{MainThmMultipleOutput}\label{thm:MainMultiple}
For every $\ell\le n / \poly(\log(n))$, there exists $m=\ell \cdot\poly(\log(n))$ and a polynomial-time computable function $\Cond: \zo^n \to \zo^m$ such that any oblivious $\ell$-local decision tree computing $\Cond$ has depth $n-o(n)$.
\end{restatable}
Thus, our depth overhead here is $n / \poly(\log(n))$, which is nearly optimal. 
Our explicit hard function $\Cond$ is a condenser for sumset sources and has the property that for any $A, B\subseteq\zo^n$ where $\abs{A}\ge \abs{B} \ge 2^{\poly(\log(n))}$, the size of the image $\Cond(A + B)$ is at least $2^{m - o(m)}$.\footnote{These objects are also referred to as a disperser but in this work, we will only refer to dispersers as functions that have non-constant output image. Sumset condensers satisfy a stronger property but for our lower bounds, this property suffices. Moreover, our explicit construction of $\Cond$ will satisfy the standard properties of a condenser.}
Once again, as far as we are aware, this is the first instance of a lower bound that uses such pseudorandom properties of the hard function.
To get a sense of parameters, we will have that $\abs{A} \gg \abs{B}$ and $m \simeq 0.99\log(\abs{A})$. 

We also emphasize that a much weaker version of our result for (non-oblivious) local decision trees, provided $d = \omega(m)$, will lead to breakthrough lower bounds against Valiant's model and hence against frontier circuit families: For instance, a lower bound of $0.01 n$ on the depth for any growing $\ell = \omega(1)$ would suffice (see \cref{subsec: connection to valiant} for further details).

While we do have excellent extractors (whose output distribution is close to uniform) for sumset sources \cite{CL22, L23}, the number of bits outputted by these constructions for $(n, k_1, k_2)$ sumset sources (see \cref{def: sumset source}), when $k_1 \gg k_2$, is proportional to $k_2$ and not $k_1$. Therefore, as an additional result, we construct a sumset condenser that can get most of the entropy out of such sumset sources. See \cref{thm: explicit sumset condenser} for a formal statement. While constructing such a condenser, we needed to use chain rule for smooth min-entropy, a natural extension of the chain rule for min-entropy; we are not aware of an explicit reference for this and provide a self contained proof that we hope proves to be useful for future works (see \cref{lem: chain rule smooth min entropy} for details).

\paragraph{Oblivious degree-2 decision trees.}

In this setting, again, we obtain maximal lower bounds.
\begin{restatable}{maintheorem}{MainThmQuadratic}\label{thm:MainQuadratic}
There exists a polynomial-time computable function $\Disp: \zo^n \to \zo$ such that any oblivious degree-$2$ decision tree computing $\Disp$ has depth at least $n - O(\log(n))$.
\end{restatable}
While there is no notion of depth overhead for these trees, we find it quite surprising that we can obtain such strong lower bounds, with such a small additive gap to $n$. 
Our hard function $\Disp$ here is a directional affine disperser which has the property that for all $a\ne 0^n$, the function $f_a(x) = \Disp(x) + \Disp(x + a)$ is an affine disperser for dimension $O(\log(n))$ (see \cref{def: directional affine disperser}). Such a disperser (with even stronger guarantees) was recently constructed by Li and Zhong \cite{LZ24}.

\subsection{Proof overview}\label{sec:overview}

In this section, we provide an overview of each of our lower bound arguments as well as the construction of the sumset condenser.

\subsubsection{Single-output oblivious local decision trees}
We sketch the proof of \cref{thm:MainSingle}. For clarity, we fix $\ell = n^{0.49}$ and prove that any $\ell$-local decision tree computing a \emph{sumset disperser} has depth $d \ge 0.99n$.
At a high level, we use a separator-based argument to find sumset structure inside the local computation.
As mentioned above, our explicit hard function $\Disp:\zo^n \to \zo$ is chosen so that for any $A,B \subseteq \zo^n$ with $\abs{A},\abs{B} \ge \poly(n)$, we have $\Disp(A+B)=\zo$.
Given a low-depth tree $T$ computing $\Disp$, we will construct $A$ and $B$ that violate this property.

Let $T$ be an oblivious $\ell$-local decision tree of depth $d=0.99n$ computing $\Disp$.
We can associate with $T$ an $\ell$-local map $L:\zo^n \to \zo^d$ and a function $h$ such that
\[
T(x)=h(L(x)) \qquad \text{for all } x\in \zo^n.
\]
We will construct sets $A,B$ with $\abs{A},\abs{B}\ge \poly(n)$ such that
\[
L(a)=L(a+b) \qquad \text{for all } a\in A,\; b\in B.
\]

\paragraph{Concluding the theorem from such $A$ and $B$.}
Partition $A$ into $A_0$ and $A_1$, where $a\in A_0$ if $h(L(a))=0$, and $a\in A_1$ otherwise. Without loss of generality, $\abs{A_0}\ge \abs{A}/2 \ge \poly(n)$. Then for every $a_0\in A_0$ and $b\in B$, $h(L(a_0+b)) = h(L(a_0)) = 0$, and, hence, $T(A_0+B)=h(L(A_0+B))=\{0\}$, contradicting that $T$ computes the sumset disperser $\Disp$.

\paragraph{Constructing $A$ and $B$.}
Let $\alpha = O(\log n)$. Order the $n$ input variables by the number of functions $L_i$ in which they appear, in non-decreasing order, and let $I\subseteq [n]$ be the set of the first $\alpha$ variables in this order. Since there are $d$ many $\ell$-local functions, the variables in $I$ appear in at most $\frac{\alpha d \ell}{n} \le \alpha \ell$ many functions $L_i$. In particular, there are at most $\alpha \ell^2 = o(n)$ variables that appear in some function $L_i$ together with a variable from $I$.
Let $S\subseteq [n]$ be the set of all such variables, note that $\abs{S}=o(n)$. And let $R = [n]\setminus (I\cup S)$ be the set of the remaining variables, satisfying $\abs{R} \ge n-\alpha-\alpha \ell^2 = n-o(n)$.

By construction, $S$ \emph{separates} $I$ from $R$: no function $L_i$ depends on variables from both $I$ and $R$ simultaneously.
Let $d_1$ be the number of output functions $L_i$ that depend only on variables from $I\cup S$, and let $Q_1:\zo^{\abs{I}+\abs{S}} \to \zo^{d_1}$ be the corresponding submap of $L$. Let $d_2 = d-d_1$, and let $Q_2:\zo^{\abs{R}+\abs{S}} \to \zo^{d_2}$ be the remaining submap, whose outputs depend only on variables from $R\cup S$. Hence, $Q_1, Q_2$ partition the output bits of $L$ and we can write $L(x) = (Q_1(x_I, x_S), Q_2(x_I, x_S))$.

Fix the variables in $S$ to $0^{\abs{S}}$, and define the restricted maps 
\begin{align*}
    Q_1' &= Q_1(\cdot,0^{\abs{S}})\colon\zo^{\abs{I}} \to \zo^{d_1} \;,\\
    Q_2' &= Q_2(\cdot,0^{\abs{S}})\colon\zo^{\abs{R}} \to \zo^{d_2} \;.
\end{align*}
Let $z^*\in \zo^{d_2}$ be a value with the largest preimage under $Q_2'$, and let
$Y(z^*) = \{\, y\in \zo^{\abs{R}} : Q_2'(y)=z^* \,\}$. We pick any arbitrary $y^*\in Y(z^*)$ and define our sets $A, B$ as follows:
Let $A = A_S\times A_I\times A_R$ and $B = B_S\times B_I\times B_R$, where 
\begin{align*}
A_S&=\{0^{\abs{S}}\}\;, & A_I&=\zo^{\abs{I}}\;, & A_R&=\{y^*\}\;;\\
B_S&=\{0^{\abs{S}}\}\;, & B_I&=\{0^{\abs{I}}\}\;, & B_R&=Y(z^*) + y^* \;.
\end{align*}

\paragraph{Verifying $L(a)=L(a+b)$.}
We claim that $L(a)=L(a+b)$ for all $a\in A$ and $b\in B$. Using the decomposition of $L$ into $Q_1$ and $Q_2$ and the fact that $A_S=B_S=\{0^{\abs{S}}\}$, it suffices to check the claim for $Q_1'$ and $Q_2'$. For our choices of $A$ and $B$, we indeed have
\begin{align*}
Q_1'(a+b)&=Q_1'(a+0^{|I|})=Q_1'(a) \;, \\
Q_2'(a+b)&=Q_2'(y^*+(y+y^*))=z^*=Q_2'(a)
\end{align*}
for a $y\in Y(z^*)$.

\paragraph{Sizes of $A$ and $B$.}
We now show $A$ and $B$ are large enough. First, $\abs{A} \ge 2^{\alpha} \ge \poly(n)$ by the choice of $\alpha=O(\log{n})$.
And second, $\abs{B} = \abs{Y(z^*)} \ge 2^{\abs{R} - d_2} \ge 2^{\abs{R} - d}$. 
From our assumption on $\abs{R}$, we infer that $\abs{B} \ge 2^{n - o(n) - d}$.
As $d = 0.99 n$, we conclude that $\abs{B} \ge 2^{0.01 n - o(n)} \ge \poly(n)$, concluding our result.

\subsubsection{Multi-output oblivious local decision trees}

We sketch the proof of \cref{thm:MainMultiple}. For clarity, we fix $\ell=n^{0.99}$ and $m=\widetilde{O}(\ell)$, and we prove that any $\ell$-local decision tree computing a \emph{sumset condenser} $\Cond\colon \zo^n \to \zo^m$ has depth $d\geq0.99n$. 

We use the following property of a sumset condenser. For any $A, B\subset \zo^n$ with $\abs{A}\geq 2^{\widetilde{\Omega}(\ell)}$ and $\abs{B} \geq 2^{\poly(\log(n))}$, it holds that $\abs{\Cond(A + B)} \ge 2^{m - o(m)}$.
We sketch our construction of an explicit sumset condenser $\Cond$ in \cref{subsec: proof overview condenser}

Let $T$ be an oblivious $\ell$-local decision tree of depth $d$ computing $\Cond$.
As before, we associate with $T$ an $\ell$-local map $L:\zo^n \to \zo^d$ and a function $h$ such that for all $x\in\zo^n$, $T(x)=h(L(x))$.
We will find large sets $A, B$ such that $L(A + B)$ has small image. This implies that  $T(A+B) = h(L(A+B))$ also has small image, contradicting the condenser property. 

\paragraph{Constructing $A$ and $B$.}
Let $z^*\in \zo^d$ be a value with the largest preimage under $L$, and let $A = \{\, a\in \zo^{n} : L(a)=z^* \,\}$. We have that $\abs{A}\geq2^{n-d}=2^{0.01n}$. On average, an input variable appears in at most $d\ell/n \le \ell$ many functions $L_i$.
For simplicity, assume that each variable appears in at most $\ell$ functions.  In this case, we let $B$ be the Hamming ball of radius $O(\poly(\log(n))$.\footnote{In the general case, we set $B$ to be an appropriate subset of this Hamming ball.}

\paragraph{Concluding the theorem from such $A$ and $B$.}
We claim that for any $y_0\in A + B$, the Hamming distance between $z^*$ and $L(y_0)$ is at most $\ell\cdot \poly(\log(n))$.
Showing this suffices to prove the theorem since then the output of $L(A + B)$ will lie within a ball of radius $r = \ell\cdot \poly(\log(n))$ around $z^*$, showing that  $\abs{T(A + B)}$ is at most $n^r = 2^{r\log(n)}\ll2^{m-o(m)}$, contradicting the condenser property.

\paragraph{Hamming distance between $z^*$ and $L(A+B)$.}
We now show that for any $y_0\in A + B$, the Hamming distance between $z^*$ and $L(y_0)$ is at most $\ell\cdot \poly(\log(n))$.
As $B$ is a Hamming ball of radius $\poly(\log(n))$, there exists some $y\in A$ such that the Hamming distance between $y$ and $y_0$ is at most $\poly(\log(n))$. Since we assumed that each input variable appears in at most $\ell$ functions, $L(y)$ and $L(y_0)$ differ in at most $\ell\cdot \poly(\log(n))$ positions, and so the Hamming distance between $z^*=L(y)$ and $L(y_0)$ is at most $\ell\cdot \poly(\log(n))$ as desired.

\subsubsection{Oblivious degree-2 decision trees}
We sketch the proof of \cref{thm:MainQuadratic}.
As mentioned earlier, our hard function $\Disp\colon\zo^n \to \zo$ will be a \emph{directional affine disperser}: for every non-zero $a\in\zo^n$ and every affine subspace $S$ of dimension at least $\Omega(\log(n))$, $\Disp(x+a)-\Disp(x)$ is not constant on $S$. 

Let $T$ be an oblivious degree-2 decision tree of depth $d = 0.99n$ computing $\Disp$. We can write $T(x) = h(Q(x))$, where $Q\colon \zo^n \to \zo^d$ is a degree-$2$ map and $h: \zo^d \to \zo$ is a function. 
We will find an affine subspace $S$ of dimension at least $n-d \ge \Omega(\log(n))$ and a direction $a\in \zo^n$ such that $Q(x) = Q(x + a)$  for all $x\in S$.
It will then follow that 
\[
T(x) = h(Q(x)) = h(Q(x+a)) = T(x+a) \qquad \text{for all } x\in S\;,
\]
contradicting the disperser property.

\paragraph{Constructing $S$ and $a$.}
It remains to construct such an affine subspace $S$ and direction $a$.
Since $d < n$, there exist distinct $u, v\in \zo^n$ with $Q(u) = Q(v)$.
We let $a = u + v$ and consider $L_a(x) = Q(x) + Q(x + a)$.
As $Q$ is a degree-$2$ map, $L_a$ must be a degree-$1$ map.
By our choice of $u$ and $v$, we see that $L_a(u) = Q(u) + Q(v) = 0^d$ and hence $0^d$ lies in the image of the affine map $L_a$.
Therefore, $S=L_a^{-1}(0^d)$ is a non-empty affine subspace of dimension at least $n-d$.  
Then, for all $x\in S$, we have $L_a(x) = Q(x) + Q(x + a) = 0^d$ and hence $Q(x) = Q(x + a)$ for all $x\in S$, as desired.

\subsubsection{Constructing explicit sumset condenser}\label{subsec: proof overview condenser}

We here sketch the construction of our explicit sumset condenser $\Cond$. Throughout this section, we assume the reader is familiar with notions of statistical distance, min-entropy, smooth min-entropy, sumset source, seeded extractor, and data processing inequality. See \cref{subsec: stat distance}, \cref{subsec: min-entropy}, \cref{subsec: sumset sources}, and \cref{subsec: seeded extractor} for the relevant definitions.

We require a condenser $\Cond: \zo^n \to \zo^m$ for $(n, k_1, k_2)$ sumset sources where $k_1 = \frac{n}{\log(\log(n))}, k_2 = \poly(\log(n)), m = 0.99 k_1$ such that for any such source $\X$, $\minH^{\eps}(\Cond(\X)) \ge m - o(m)$ for $\eps = 0.01$.

To construct our explicit sumset condenser $\Cond$ with these desired properties, we rely on previous constructions of excellent extractors $\Ext: \zo^n \to \zo^t$ for $(n, k, k)$ sumset sources for any $k\ge \poly(\log(n))$. This extractor will have the property that $(\X_1, \Ext(\X)) \approx_{n^{-0.01}} (\X_1, \U_t)$ where $t = k^{0.01}$ \cite{CL22}. 

Ideally, given such $\Ext$, we would have liked to output $\sExt(\X_1, \Ext(\X))$ where $\sExt$ is a seeded extractor which can output $0.99 k_1$ bits; this is because by data processing inequality, $\sExt(\X_1, \Ext(\X))  \approx_{\eps}\sExt(\X_1, \U_t) \approx_{\eps} \U_m$, concluding the construction. Such techniques have been deployed previously in constructions of two source extractors \cite{L16}. However, since we only have access to $\X = \X_1 + \X_2$, we are unable to do so. Therefore, we exploit the fact that $\X_2$ has much smaller entropy and construct a condenser.

To construct our condenser $\Cond$, we use standard construction of a linear seeded extractor $\LExt: \zo^n \times \zo^t\to \zo^m$ for min-entropy $0.999 k_1$ distributions where $m = 0.99 k_1$. This has the property that $\LExt(\cdot, y)$ is a linear function of the input for all $y\in \zo^t$.\footnote{Linear seeded extractors have proven to be incredibly helpful in constructions of seedless extractors for affine sources, and sumset sources.}
With this, our construction of $\Cond$ will be simply $\Cond(\X) = \LExt(\X, \Ext(\X))$.

Using linearity of the seeded extractor, we can write $\LExt(\X, \Ext(\X)) = \LExt(\X_1, \Ext(\X)) + \LExt(\X_2, \Ext(\X))$.
As argued earlier, we will have that $\LExt(\X_1, \Ext(\X))$ is statistically close to $\U_m$.
To finish our analysis, we observe that since $k_2, t \ll m$, the size of support of $\LExt(\X_2, \Ext(\X))$ is very small.
Using chain rule for smooth min-entropy, we will be able to show that for most fixings of $\LExt(\X_2, \Ext(\X))$, the distribution $\LExt(\X_1, \Ext(\X))$ will have high smooth min-entropy. For every such ``good'' fixing $z_2$ of $\LExt(\X_2, \Ext(\X))$, our output distribution will be $z_2 + \LExt(\X_1, \Ext(\X)) \,|\, (\LExt(\X_2, \Ext(\X)) = z_2)$. As $z_2$ is a fixed string and $\LExt(\X_1, \Ext(\X))\, |\, (\LExt(\X_2, \Ext(\X)) = z_2)$ has high smooth min-entropy, we conclude that our overall output distribution will have high smooth min-entropy, as desired.

As mentioned earlier, we are not aware of an explicit reference for the chain rule for smooth min-entropy; hence we provide a self-contained proof of this. See \cref{lem: chain rule smooth min entropy} for further details.

\paragraph{Organization}
We explicitly map out the connections to various circuit models in \cref{subsec: Connections to other circuit models}.
We include the necessary preliminaries in \cref{sec:prelims}.
In \cref{sec: seperator based lower bound} we prove our lower bound for single-output oblivious local decision trees. In \cref{sec: noise sensitivity based lower bound} we prove the lower bound for multi-output oblivious local decision trees, and in \cref{sec:condenser} we construct the required sumset condenser. Finally, in \cref{sec: oblivious deg 2 decision tree} we show the lower bound for oblivious degree-2 decision trees .
We lay out many open directions in \cref{sec: conclusion and open problems}.

\section{Connections to other circuit complexity models}
\label{subsec: Connections to other circuit models}

\subsection{Depth-3 circuits.}
The closest connection to local decision trees is with $\Sigma_3^{\ell}$ circuits---these are $\OR\circ\AND\circ\OR_{\ell}$ circuits where $\ell$ is the bottom fan-in. Equivalently, these circuits compute an $\OR$ of $\ell$-CNFs. We can easily see that if a function $f$ can be computed by an $\ell$-local decision tree of depth $d$, then it can also be computed by a $\Sigma_3^{\ell}$ circuit of size $2^{d + \ell}$. This means that lower bounds against these depth-$3$ circuits yield lower bounds against $\ell$-local decision trees. For a lot of other circuit models, the task of proving non-trivial lower bounds can be reduced to proving near maximal lower bounds against depth-$3$ circuits. One of our main conceptual contributions is that all known reductions from other circuit models to depth-3 circuits can instead be made to $\ell$-local decision trees where a size $\ge s$ lower bound requirement against depth-3 circuit often translates to that of proving depth $\ge \log(s)$ lower bound against $\ell$-local decision trees. We will go over various computational models below specifying these exact connections.

We further believe that proving lower bounds against $\ell$-local decision trees can be a much simpler task than proving lower bounds against $\Sigma_3^{\ell}$ circuits. 
For instance, for $\ell$-local decision trees, it is trivial to see that computing parity requires depth $n / \ell$ and that this is tight. However, proving this for $\Sigma_3^{\ell}$ circuits was an open problem and only after extensive efforts, was the tight bound of $2^{n / \ell}$ obtained~\cite{PPZ99}. Another example illustrating this is that Chistopolskaya and Podolskii showed that a $2$-local decision tree computing Majority must have depth at least $n - o(n)$~\cite{CP22}. This is in contrast to the case of $\Sigma_3^2$ circuits where the size of optimal circuit computing Majority is $2^{n/2}\cdot \poly(n)$ \cite{HJP95}, yielding a trivial lower bound of $n / 2 - o(n)$ on the depth of any $2$-local decision tree computing Majority.

\subsection{Valiant's model.}
\label{subsec: connection to valiant}
Valiant showed that to prove super-linear circuit lower bounds against log-depth circuits as well as against series-parallel circuits (these are still unconquered frontiers in circuit complexity), it suffices to prove strong lower bounds against a new model of computation that he introduced---Valiant's model~\cite{V77}.
He further showed that to prove lower bounds against his model, it suffices to prove strong lower bounds against $\Sigma_3^{\ell}$ circuits.

We observe that $\ell$-local decision trees sandwich perfectly in between these depth-3 circuits and Valiant's model---that such depth-3 circuit lower bounds yield lower bounds against $\ell$-local decision trees; and that those exact lower bounds against $\ell$-local decision trees yield the same lower bound against Valiant's model, yielding super-linear circuit lower bounds.
We define Valiant's model and provide exact parameter translation in \cref{subsec: Valiant}.

From these connections, we obtain the following consequences from lower bounds on local decision trees.
First, finding an explicit function that requires depth $\Omega(n)$ against $\ell$-local decision trees for any $\ell \in \omega(1)$ yields super-linear lower bounds against series-parallel circuits. And finding an explicit function that requires depth $\omega(\frac{n}{\log(\log(n))})$ against $\ell = n^{\eps}$-local decision trees, for any constant $\eps > 0$, yields super-linear lower bounds against log-depth circuits. 
In fact, for these lower bounds, it suffices to obtain lower bound against $\ell$-local decision trees outputting many---say, $m$ bits, provided that $m \le o(d)$. 

In fact, above, instead of proving lower bounds against $\ell$-local decision trees, it suffices to prove all these lower bounds against semi-oblivious local decision trees: $\ell$-local decision trees where all nodes at the same level are (potentially different) functions over the same set of $\ell$ variables (see \cref{sec:decTrees} for formal treatment).

\subsection{Unrestricted-depth circuit lower bounds.}
The work~\cite{GKW21} showed that to improve the state of the art lower bounds for unrestricted-depth circuits (circuits with arbitrary fan-in 2 gates and no depth restriction) from current state of the art of $3.1n - o(n)$ to $3.9n-o(n)$, it suffices to obtain a lower bound of size $2^{n - o(n)}$ on $\Sigma_3^{16}$ circuits. We observe from their proof that it suffices to obtain a size lower bound of $2^{n - o(n)}$ against $16$-local decision trees, a simpler task. Also from earlier discussion, we have that $2^{n - o(n)}$ lower bounds on $\Sigma_3^{16}$ circuits yield size lower bound of $2^{n - o(n)}$ against $16$-local decision trees, demonstrating that the latter is an easier task.

\subsection{Dispersers for quadratic varieties.}
The work~\cite{GK16} showed that to improve these unrestricted-depth circuit lower bounds from $3.1n - o(n)$ to $3.11n$, it suffices to explicitly construct the following pseudorandom function---a disperser for degree-$2$ varieties of min-entropy $o(n)$ (see \cref{def: disperser deg d variety sources}). 
We observe that the task of proving lower bounds against degree-2 decision trees is an easier task than that of constructing a degree-$2$ variety disperser (see \cref{subsec: variety sources} for further details). 
In~particular, any degree-2 decision tree computing such a disperser with min-entropy requirement $o(n)$ would require depth $n - o(n)$. 
Since constructing such dispersers seems to be beyond the reach of current techniques, we hope that the tools developed to obtain lower bounds against degree-2 decision trees can help inspire improved constructions of such dispersers.

\subsection{Oblivious Branching programs.}

An oblivious branching program is a layered branching programs where all nodes in the same layer read the same variables. It has been observed by previous works that if a function can be computed by a width-$w$ length-$t$ oblivious branching program, then it can be computed by an oblivious $\ell$-local decision tree of depth $w\log(w)\cdot t / \ell$~\cite{HR15, LRT22}. Hence, lower bounds against $\ell$-local decision trees result in lower bounds against such branching programs.

\section{Preliminaries}\label{sec:prelims}
For a positive integer~$n$, by $[n]$ we denote the set $\{1,\ldots,n\}$.  
All logarithms are base~$2$, i.e., $\log(2^n)=n$. Let $f\colon\zo^n\to\zo^m$ be a function. For a set $X\subseteq\zo^n$, $f(X)$ denotes the image of $X$ under $f$, and for $y\in\zo^m$, $f^{-1}(y)$ denotes the preimage of $y$ under~$f$. We will use boldface font to denote random variables. We will sometimes abuse notation and interchangeably use the terms random variables, distributions, and sources. For a distribution $\X$, we write $\supp(\X)$ to denote its support. We will use $\U_n$ to denote the uniform distribution over $\zo^n$.

For an integer $\ell\geq 2$, an $\ell$-uniform hypergraph is a hypergraph where each edge contains exactly $\ell$ distinct vertices. For a subset of vertices $S$ of a hypergraph $H$ over vertex set $V$, the neighborhood $\Nbr(S)$ is the set of all vertices adjacent to a vertex from $S$, $\Nbr(S)=\{v\in V\colon \exists s\in S, \exists e \in H, \text{ s.t. }v,s\in e\}$.

\subsection{\texorpdfstring{$\cF$-decision trees}{F-decision trees}}\label{sec:decTrees}

We first define the general notion of $\cF$-decision trees and will specialize it to local and low-degree decision trees.
\begin{definition}[$\cF$-decision tree]
Let $n, m\in \N$ and let $\cF$ be a class of Boolean functions over $n$ variables.
An \emph{$\cF$ decision tree} $T: \zo^n \to \zo^m$ is a binary tree where each internal node is marked by a Boolean  function $f\in \cF$ over the $n$ input variables with two outgoing edges marked with $0$ or $1$ (the outcome of the function), and each leaf node marked by an outcome from $\zo^m$.

The depth of an $\cF$-decision tree is the depth of the underlying binary tree.
\end{definition}

We will study two specific instantiations of these trees: (i) \emph{$\ell$-local decision trees} where $\cF$ is the family of all $\ell$-local Boolean functions; and (ii) \emph{degree-$d$ decision trees} where $\cF$ is the family of multivariate degree-$d$ polynomials over $\F_2$.

We specialize this definition further to define the oblivious version of these trees.
\begin{definition}[Oblivious $\cF$-decision tree]
Let $\cF$ be a class of Boolean functions and let $T$ be an $\cF$-decision tree.
We say $T$ is an \emph{oblivious $\cF$-decision tree} if all nodes at the same level of the tree are marked by the exact same function from $\cF$.
\end{definition}

We can characterize oblivious $\cF$-decision trees using a composition-based definition as well. In~particular, for any oblivious $\cF$-decision tree $T: \zo^n \to \zo^m$ of depth $d$, we can associate an $\cF$-map $P: \zo^n \to \zo^d$ (where each output function $P_i\in \cF$) and a function $h: \zo^d \to \zo^m$ (of arbitrary complexity) such that $T(x) = h(P(x))$.

Here, at each node at level $i$, the query of the decision tree is the function $P_i$, and the label of each leaf is dictated by the function $h$.

We will also consider semi-oblivious $\ell$-local decision trees defined as follows.
\begin{definition}[Semi-Oblivious local decision tree]
Let $T$ be an $\ell$-local decision tree.
We say that $T$ is a \emph{semi-oblivious $\ell$-local decision tree} if all nodes at level $i$ are marked by (possibly different) $\ell$-local functions of the same $\ell$ variables.
\end{definition}

\subsection{Valiant's model}\label{subsec: Valiant}

We define Valiant's model as follows:
\begin{definition}[Valiant's model~\cite{V77}]
For $\ell, m, n, t\in \N$, we say a function $f: \zo^n \to \zo^m$ is computable in Valiant's model with $t$ intermediate bits and locality $\ell$ if the following holds.
There exist $t$ intermediate functions $h_1, \dots, h_t$ where each $h_i\in\zo$ is a function of the previous bits $h_1, \dots, h_{i-1}$ and $\ell$ input bits.
Furthermore, for all $i\in [m]$, the output bit $f_i$ is a function of the intermediate bits $h_1, \dots, h_t$ and $\ell$ input bits.
\end{definition}

Valiant~\cite{V77} showed that lower bounds in this model could be used to prove breakthrough circuit lower bounds. In particular, a lower bound of $t\geq\Omega(n)$ for any $\ell \in \omega(1)$ would imply a super-linear lower bound on the size of series-parallel circuits. Similarly, a lower bound of $t\geq\omega\left(\frac{n}{\log(\log(n))}\right)$ for $\ell=n^\eps$ would imply a super-linear lower bound on the size of log-depth circuits.

Below we show that if $f$ is computable in Valiant's model, then it is also computable by semi-oblivious local decision trees.
\begin{claim}\label{claim: valiant to local dt}
Let $\ell, m, n, t\in \N$, and let $f: \zo^n \to \zo^m$ be computable by Valiant's model with~$t$ intermediate bits with locality $\ell$.
Then, there exists an $\ell$-local semi-oblivious tree of depth $t + m$ computing $f$.
\end{claim}

\begin{proof}
Let the intermediate bits be $h_1, \dots, h_t$.
For $i\in [t]$, we let the nodes at level $i$ compute the output of the intermediate bit $h_i$. This is possible since for a node at level $i$, the path from root to that node fixes the values of all previous intermediate bits, leaving $h_i$ to be an $\ell$-local function of the input.
Then for $i\in [m]$, we let the nodes at level $t + i$ compute the output of bit $f_i$.
Again, this is possible since after fixing the intermediate bits, $f_i$ is an $\ell$-local function of the input.
Finally, we label each leaf with the path taken from level $t$ to the leaf.
\end{proof}

This implies that for a function $f: \zo^n \to \zo^m$, if any semi-oblivious $\ell$-local decision tree computing $f$ has depth $d$, then it requires at least $d - m$ intermediate bits to be computed in Valiant's model.

In particular, for any $\ell=\omega(1)$, a lower bound of $d\geq\Omega(n)$ on the depth of semi-oblivious $\ell$-local decision tree would lead to a super-linear lower bound for series-parallel circuits.

\subsection{Statistical distance}\label{subsec: stat distance}

\begin{definition}[Statistical distance]
Let $\X, \Y\sim \Omega$. The \emph{statistical distance} between $\X$ and $\Y$ is defined as 
\[
    \abs{\X - \Y} = \max_{S\subset \Omega} \left(\Pr[\X\in S] - \Pr[\Y\in S]\right) = \frac{1}{2}\sum_{s\in \Omega} \abs{\Pr[\X = s] - \Pr[\Y = s]}\;.
\]
\end{definition}
If $\abs{\X - \Y} \le \eps$, then we say that $\X$ is $\eps$-close to $\Y$, and denote this by $\X \approx_{\eps} \Y$. We will use the fact that statistical distance satisfies the triangle inequality.
We will also use the following standard result regarding statistical distance.
\begin{lemma}[Data processing inequality]\label{lem: data processing inequality}
For any random variables $\X, \Y\sim \Omega_1$ and any function $f: \Omega_1 \to \Omega_2$, 
\[
    \abs{\X - \Y} \ge \abs{f(\X) - f(\Y)} \;.
\]
\end{lemma}

\subsection{Min-entropy}\label{subsec: min-entropy}

We will use the standard notion of min-entropy to measure the amount of randomness in a random variable.
\begin{definition}\label{def: min-entropy}
For a random variable $\X\sim \Omega$, we define its \emph{min-entropy} as 
\[
\minH(\X) = \min_{x\in \Omega} \log(1 / \Pr[\X = x])\;.
\]
\end{definition}
We say $\X$ is an \emph{$(n, k)$ source} if $\X\sim \zo^n$ and $\minH(\X) = k$.
We will also use the following notion of smooth min-entropy. We say that $\X$ is a \emph{flat $(n,k)$ source} if $\X$ is a uniform distribution over a set $S\subseteq\zo^n$ of size $|S|=2^k$ (provided that $2^k\in\N$).
\begin{definition}
For a random variable $\X\sim \Omega$ and $\eps > 0$, we define its \emph{smooth min-entropy} with error $\eps$ as 
\[
\minH^{\eps}(\X) = \max_{\Y\sim \Omega: \abs{\X - \Y}\le \eps}\{\minH(\Y)\} \;.
\]
\end{definition}

We will use the chain rule for min-entropy.
\begin{lemma}[Chain rule for min-entropy~\cite{MW97}]\label{lem: chain rule min entropy}
For any random variables $\X\sim \Omega_X$ and $\Y\sim \Omega_Y$ and $\eps > 0$, we have
\[
    \Pr_{y\sim \Y}[\minH(\X | \Y = y) \ge \minH(\X) - \log(\abs{\supp(\Y)}) - \log(1/\eps)] \ge 1 - \eps \;.
\]
\end{lemma}

We will rely on the chain rule for smooth min-entropy. Since we are not aware of an explicit reference for the following lemma, we supply its proof using standard ideas.
\begin{lemma}[Chain rule for smooth min-entropy]\label{lem: chain rule smooth min entropy}
For any random variables $\X\sim \Omega_X$ and $\Y\sim \Omega_Y$ and $\eps, \gamma, \delta > 0$, we have
\[
\Pr_{y\sim \Y}\left[\minH^{\gamma}(\X | \Y = y) \ge \minH^{\eps}(\X) - \log(\abs{\supp(\Y)}) - \log(1/\delta)\right] \ge 1 - (\eps + \delta + (2\eps / \gamma)) \;.
\]
\end{lemma}
\begin{proof}
Let $\X' \sim \Omega_X$ be such that $\minH(\X') = \minH^{\eps}(\X)$ and $\abs{\X - \X'} \le \eps$. Note that we can always pick such an $X'$ satisfying $\supp(\X)\subset \supp(\X')$.
We then define $\Y'\sim \Omega_Y$ such that for all $x\in \supp(\X)$ and $y\in \supp(\Y)$, $\Pr[\Y' = y | \X' = x] = \Pr[\Y = y | \X = x]$ 
Also for $x\in \supp(\X')\setminus \supp(\X)$, we let $(\Y' | \X' = x)$ be uniformly distributed over $\supp(\Y)$. This finishes our description of the distributions $\X'$ and $\Y'$, satisfying $\supp(\X)\subset \supp(\X')$ and $\supp(\Y') = \supp(\Y)$.
We also  observe that by definition of these random variables, $\abs{(\X, \Y) - (\X', \Y')} \le \abs{\X-\X'} \leq \eps$. From the data processing inequality (\cref{lem: data processing inequality}), we obtain that $\abs{\Y - \Y'}\le \eps$.

Let $\Bad = \{y\in \supp(\Y'): \minH(\X' | \Y' = y) \le \minH(\X') - \log(|\supp(\Y')|) - \log(1 / \delta)\}$.
Applying the chain rule for min-entropy (\cref{lem: chain rule min entropy}) to $\X'$ and $\Y'$, we infer that $\Pr[\Y'\in \Bad] \le \delta$.
Hence, we conclude that $\Pr_{y\in \Y}[y\in \Bad]\le \eps + \delta$.

Next, we define $\Far = \{y\in \supp(\Y): \left|\,(\X | \Y = y) - (\X' | \Y' = y)\,\right| \ge \gamma\}$.
We claim that $\E_{y\in \Y} \left[\abs{\,(\X | \Y = y) - (\X' | \Y' = y)\,}\right] \le 2\eps$.
Indeed, we compute that
\begin{align*}
 \E_{y\sim \Y} &\left[\abs{\,(\X | \Y = y) - (\X' | \Y' = y)\,}\right]\\
& = \sum_{x\in \Omega_X, y\in \Omega_Y} \Pr[\Y = y] \cdot \frac12 \abs{\Pr[\X = x | \Y = y] - \Pr[\X' = x | \Y' = y]}\\
& \le \abs{\Y - \Y'} + \frac12\sum_{x\in \Omega_X, y\in \Omega_Y} \abs{\Pr[\Y = y]\Pr[\X = x | \Y = y] - \Pr[\Y' = y]\Pr[\X' = x | \Y' = y]}\\
& \le \eps + \abs{\,(\X,\Y)-(\X',\Y')\,} \\
&= 2\eps \;.
\end{align*}
Hence, the claim holds.
We now apply Markov's inequality to infer that 
$\Pr_{y\sim \Y}[y\in \Far]\le \frac{2\eps}{\gamma}$.

Therefore, $\Pr_{y\sim\Y}[y\in \Bad\cup\Far]\le \eps + \delta + \frac{2\eps}{\gamma}$.
As $\supp(\Y') = \supp(\Y)$ and $\minH(\X') = \minH^{\eps}(\X)$, we infer that for all $y\in \supp(\Y)\setminus (\Bad\cup \Far)$, it holds that $(\X | \Y = y)$ has statistical distance at most $\gamma$ from a distribution with min-entropy at least $\minH^{\eps}(\X) - \log(\abs{\supp(\Y)}) - \log(1/\delta)$, as desired.
\end{proof}

\subsection{Variety sources}\label{subsec: variety sources}

Below we define degree-$d$ variety sources over $\F_2$.
\begin{definition}[Degree-$d$ variety sources]
For $d, n\in \N$, a degree-$d$ variety source $\X\sim \zo^n$ is associated with a polynomial map $P = (p_1, \dots, p_t)\colon \F_2^n \to \F_2^t$ for some $t\in \N$, where each $p_i$ is a multilinear degree-$d$ polynomial of the input. The source $\X$ is uniform over the set of common zeroes of this map: $V = \{x\in \F_2^n: P(x) = 0^t\}$.
\end{definition}

We define dispersers for these sources as follows.
\begin{definition}[Dispersers for degree-$d$ variety sources]\label{def: disperser deg d variety sources}
For $d, k, n\in \N$, a function $\Disp: \zo^n \to \zo$ is a disperser for degree-$d$ variety sources with min-entropy $k$, if for all degree-$d$ variety sources $\X$ of min-entropy at least $k$, it holds that $\Disp(\X) = \zo$.
\end{definition}

In the work of \cite{GK16}, it was shown that dispersers for degree-$2$ variety sources for min-entropy $o(n)$ require unrestricted-depth circuits of size at least $3.11 n$, which would improve upon the current state of the art lower bound of $3.1n - o(n)$ for an explicit function~\cite{LT22}. However, the current best dispersers for degree-$2$ variety sources require min-entropy at least $(1 - \gamma)n$ for a small absolute constant $\gamma > 0$ \cite{LZ19}.

We next show that dispersers for degree-$t$ variety sources require large depth to be computed by a degree-$t$ decision tree, thereby establishing that proving lower bounds for degree-$t$ decision trees is an easier task than constructing said dispersers.
\begin{claim}
Let $\Disp: \zo^n \to \zo$ be a disperser for degree-$t$ variety sources with min-entropy requirement $k$.
Then, any degree-$t$ decision tree computing $\Disp$ requires depth at least $n - k + 1$.
\end{claim}

\begin{proof}
Suppose there exists a degree-$t$ decision tree $T$ of depth $d\le n-k$ computing $\Disp$.
By an averaging argument, there exists a leaf node $L$ such that at least $2^{n-d}$ inputs from $\F_2^n$ end up at that leaf.
For each node $v$ in the path from root to $L$, let $p_v$ be the degree $t$ polynomial associated with it.
Let $q_v$ be the polynomial that equals $p_v$ if the path from root to $L$ took the zero branch at node $v$, and let $q_v = p_v + 1$ otherwise.
We observe that the set of inputs that end up at $L$ is exactly described by the variety where each of the polynomials $q_v$ is set to $0$.
Hence, we have that $T$ is constant over a variety of size at least $2^{n - d} \ge 2^k$, contradicting the fact that it is a disperser for such sources.
\end{proof}

\subsection{Affine dispersers and directional affine dispersers}

We define affine dispersers as follows.
\begin{definition}[Affine Disperser]
For $n, k\in \N$, a function $\Disp: \zo^n \to \zo$ is an \emph{$(n, k)$ affine disperser} if for all affine subspaces $S$ of dimension at least $k$, it holds that $\Disp(S) = \{0, 1\}$.
\end{definition}

We will also require the following generalization of affine dispersers.
\begin{definition}[Directional Affine disperser]
\label{def: directional affine disperser}
For $n, k\in \N$, a function $\Disp: \zo^n \to \zo$ is an \emph{$(n, k)$ directional affine disperser} if for all directions $a\in \zo^n\setminus\{0^n\}$, it holds that $\Disp(x) + \Disp(x + a)$ is an $(n, k)$ affine disperser.
\end{definition}

We will utilize the following explicit construction of a directional affine disperser.
\begin{theorem}[{\cite[Theorem~1.4]{LZ24}}]\label{thm: explicit dir affine disperser}
There exists a constant $C_0>1$ and a polynomial-time computable $(n,k)$ directional affine disperser $\Disp\colon\{0,1\}^n\to\{0,1\}$ for any $k\geq C_0\log{n}$.\footnote{They obtain a stronger variant of these objects, namely extractors; but we only need dispersers for our applications.\label{footnote: extract not disperse}}
\end{theorem}

\subsection{Sumsets and sumset sources}\label{subsec: sumset sources}

For $x,y\in\zo^n$, $x+y$ denotes the bitwise xor of $x$ and $y$. For sets $A, B\subseteq\zo^n$, we define the sumset $A+B=\{a+b\colon a\in A, b\in B\}$. Similarly, for $x\in\zo^n$ and $A\subseteq\zo^n$, $x+A$ denotes the set $x+A=\{x+a\colon a\in A\}$.

We define sumset sources as follows.
\begin{definition}\label{def: sumset source}
We say $\X\sim \zo^n$ is a \emph{$(n, k_1, k_2)$ sumset source} if there exist two independent sources $\X_1, \X_2\sim \zo^n$ with $\minH(\X_1) = k_1, \minH(\X_2) = k_2$ such that $\X = \X_1 + \X_2$.
\end{definition}
We now define sumset dispersers.

\begin{definition}[Sumset Disperser]
For $n, k\in \N$, a function $\Disp: \zo^n \to \zo$ is an \emph{$(n, k)$ sumset disperser} if for all $(n, k, k)$ sumset sources $\X$, it holds that $\supp(\Disp(\X)) = \{0, 1\}$.
\end{definition}

We will use the following explicit construction of a sumset disperser.
\begin{theorem}[{\cite[Theorem~1.7]{L23}}]\label{thm:L23}
There exists a constant $C>1$ and a polynomial-time computable $(n,k)$ sumset disperser $\Disp\colon\{0,1\}^n\to\{0,1\}$ for any $k\geq C\log{n}$.\footref{footnote: extract not disperse}
\end{theorem}

We also need to define a stronger object---sumset extractors.
\begin{definition}[Sumset Extractor]
For $n, k\in \N$, a function $\Ext: \zo^n \to \zo^m$ is an \emph{$(n, k)$ sumset extractor} with error $\eps$ if for all $(n, k, k)$ sumset sources $\X$, it holds that $\Ext(\X) \approx_{\eps}\U_m$.
We furthermore say $\Ext$ is strong in the first source if for all sumset sources $\X = \X_1 + \X_2$, it holds that $(\X_1, \Ext(\X)) \approx_{\eps} (\X_1, \U_m)$.
\end{definition}

We will use the following construction of a sumset extractor.
\begin{theorem}[\cite{CL22}]\label{thm: sumset extractor}
There exist universal constants $C > 1, \delta > 0, \gamma > 0$ such that for all $n, k, m\in \N$ where $k\ge \log^C(n)$ and $m = k^{\delta}$, there exists an explicit function $\Ext: \zo^n \to \zo^m$ that is an $(n, k)$-sumset extractor with error $n^{-\gamma}$ that is strong in the first source.\footnote{It is not explicitly mentioned in~\cite{CL22} that the sumset extractor $\Ext$ is strong in the first source. However, this is immediate from inspecting their proof.}
\end{theorem}

We lastly define a relaxation of extractors--condensers.
\begin{definition}[Sumset condenser]
A function $\Cond: \zo^n \to \zo^m$ is a \emph{$(k_1, k_2, k_{\out})$ sumset condenser} with error $\eps$ if for all $(n, k_1, k_2)$ sumset sources $\X$, it holds that $\minH^{\eps}(\Cond(\X)) \ge k_{\out}$.
\end{definition}
In \cref{sec:condenser}, we will construct an explicit condenser for sumset sources.

\subsection{Seeded extractors and linear seeded extractors}\label{subsec: seeded extractor}

\begin{definition}[Seeded extractor]\label{def: lin seeded extractor}
For $n, k, d, m\in \N$ and $\eps > 0$, we say $\Ext: \zo^n \times \zo^d \to \zo^m$ is a \emph{seeded extractor} for min-entropy $k$ with error $\eps$ (or $(k, \eps)$-seeded extractor in short) if for all $(n, k)$ source $\X$:
\[
    \Ext(\X, \U_d) \approx_{\eps} \U_m \;.
\]
We say $d$ is the seed length of $\Ext$.
Furthermore, we say $\Ext$ is a \emph{linear seeded extractor} if for all $y\in \zo^d$, $\Ext(\cdot, y)$ is a linear map.
\end{definition}

We will utilize the following linear seeded extractor.
\begin{theorem}[\cite{T01, RRV02}]\label{thm: trevisan's extractor}
For every $n, k, m \in \N$ and $\eps > 0$, with $m\le k\le n$, there exists an explicit linear seeded extractor $\LExt: \zo^n \times \zo^d \to \zo^m$ for min-entropy $k$ and error $\eps$ where $d = O\left(\frac{\log^2(n/\eps)}{\log(k / m)}\right)$.
\end{theorem}

\section{Single-output oblivious local decision trees}\label{sec: seperator based lower bound}

In this section, we prove lower bounds for oblivious local decision trees with a single output $(m = 1)$ using separator arguments. We show that any such tree computing a sumset disperser must have large depth, thereby establishing \cref{thm:MainSingle}.

In \cref{sec:singleOutputProof}, we will prove the following technical version of \cref{thm:MainSingle}.
\begin{restatable}[Technical version of \cref{thm:MainSingle}]{theorem}{TechnicalSingle}\label{theorem: uniform decision tree theorem techinical}
Let $n, d, \ell, k\in \N$ be such that 
\(d(1+\ell^2(k+1)/n)+k\leq n\).
Let $\Disp\colon \zo^n\to\zo$ be an $(n, k)$ sumset disperser.
Then any oblivious $\ell$-local decision tree computing $\Disp$ requires depth at least $d+1$.
\end{restatable}

Now we derive \cref{thm:MainSingle} by choosing appropriate parameters in \cref{theorem: uniform decision tree theorem techinical} and instantiating it with the explicit sumset disperser construction from \cite{L23} (\cref{thm:L23}).

\MainThmSingleOutput*
\begin{proof}
Let $\Disp\colon \zo^n\to\zo$ be the $(n,k)$ sumset disperser from \cref{thm:L23}, where $k=C_0\log n$. Let $C$ be an arbitrary constant $C>C_0$. Note that for $\ell\geq n / \sqrt{C_0 \log(n)}$, the theorem statement holds trivially. Therefore, in the following we assume that $\ell< n / \sqrt{C_0 \log(n)}$. Let 
\begin{align*}\label{eq:d}
d = \frac{n(n-k)}{n+(k+1)\cdot \ell^2} \;.
\end{align*}
First we note that $d\geq 1$ for every $\ell< n / \sqrt{C_0 \log(n)}$ and all large enough~$n$. Next we check that our choice of $d$ satisfies the requirements of \cref{theorem: uniform decision tree theorem techinical}.
Therefore, applying \cref{theorem: uniform decision tree theorem techinical} with these parameter choices yields a lower bound of $d+1$ on the depth of any $\ell$-local decision tree computing $\Disp$. 

To conclude the desired bound on the depth of $\ell$-local decision trees, we note that
for all $C>C_0$ and $k=C_0\log(n)$, the following holds for all large enough values of $n$,
\[
d = \frac{n(n-k)}{n+(k+1)\cdot \ell^2} > \frac{n^2}{n+C\ell^2\log(n)} \;.
\]

The second part of the theorem follows from the first one by observing that if $\ell=o(\sqrt{n/\log(n)})$, then
$C\ell^2\log(n)=o(n).$
\end{proof}

\subsection{Proof of \texorpdfstring{\cref{theorem: uniform decision tree theorem techinical}}{the lower bound}}\label{sec:singleOutputProof}
Our theorem follows from the following lemma, which identifies additive structure in arbitrary $\ell$-local maps.

\begin{lemma}[Sumsets in Local Maps]
\label{lemma: uniform decision tree key lemma}
For all $d, \ell, \alpha\leq n$ and all $\ell$-local maps $L: \zo^n \to \zo^d$, there exist sets $A, B\subseteq \zo^n$ such that for all $a\in A, b\in B$, we have that $L(a) = L(a + b)$. Furthermore, $\abs{A} \ge 2^{\alpha}, \abs{B} \ge 2^{n - d - \alpha d\ell^2 /n}$.
\end{lemma}

\begin{proof}
Let $L=(L_1,\ldots,L_d)$, where each $L_i\colon\zo^n\to\zo$ is an $\ell$-local function. 
Consider an $\ell$-uniform hypergraph $H$ over vertex set $[n]$ with $d$ edges defined as follows. For each $i\in [d]$, let $v_1, \dots, v_{\ell}$ be the variables used in the function $L_i$ (if $L_i$ depends on $t<\ell$ variables, then we add arbitrary $\ell-t$ variables to the set). Then the hypergraph $H$ contains the edge $(v_1, \dots, v_{\ell})$. 

We claim there exists a set $I\subseteq [n]$ of size $|I|=\alpha$ such that $\abs{\Nbr(I) \cup I} \le \alpha d \ell^2/n$.
For this, note that the expected number of edges a random vertex appears in is at most $d\ell/n$. Consider a random set of $\alpha$ vertices of the hypergraph.  By linearity of expectation, the expected number of edges that these vertices appear in, is at most $\alpha d \ell /n $. This implies there exists a set $I$ of size $\alpha$ such that the number of edges containing at least one vertex from $I$ is at most $\alpha d \ell /n$.
For this set $I$, it holds that $\abs{\Nbr(I) \cup I}\le \alpha d \ell^2 /n$, as desired.

Let the separator set $S$ be the set of all neighbors of~$I$ without the vertices from~$I$, i.e., $S = \Nbr(I)\setminus I$. Let us denote the remaining vertices by $R = [n] \setminus (S\cup I)$.
Let $\sigma = \abs{S}$ and let $\rho = \abs{R}$. We remark that $\abs{R}\geq n - \abs{\Nbr(I) \cup I}\geq n - \alpha d \ell^2/n$.

Note that by construction, no edge contains a vertex from~$I$ and a vertex from~$R$ simultaneously. Therefore, for every edge $e\in H$, either  $e\subseteq I\cup S$ or $e\subseteq R\cup S$. Let $Q_1\colon \zo^{\alpha+\sigma}\to \zo^{d_1}$ be the set of functions in~$L$ that depend on the variables from $I \cup S$, and $Q_2\colon \zo^{\rho + \sigma}\to \zo^{d_2}$ be the remaining functions in~$L$ (that depend on the variables from $R \cup S$) for some $d_1 + d_2 = d$. Now, up to a permutation of the coordinates, $L(x) = (Q_1(x_I, x_S), Q_2(x_R, x_S))$.

Using this decomposition of $L$, we will be able to define our desired sets $A$ and $B$. To do that, we will fix the variables in the separator $S$ to be $0^{\sigma}$, and then carefully define our sets $A, B$.
First, for each $z\in \zo^{d_2}$, define 
\[
Y(z) = \{y\in \zo^{\rho}: Q_2(y, 0^{\sigma}) = z\}\;.
\]
Let $z^* = \argmax_z \abs{Y(z)}$, the outcome with the largest preimage with respect to $Q_2(\cdot, 0^{\sigma})$. Since $z^*$ is the outcome with the largest preimage size, by an averaging argument, we obtain that $\abs{Y(z^*)} \ge 2^{\rho - d_2}\ge 2^{\rho - d}$.
Fix $y^*\in Y(z^*)$ to be arbitrary.

We are now ready to define our sets $A, B$.
Let $A = A_S\times A_I\times A_R$ and $B = B_S\times B_I\times B_R$, where 
\begin{align*}
A_S&=\{0^{\sigma}\}\;, & A_I&=\zo^{\alpha}\;, & A_R&=\{y^*\}\;;\\
B_S&=\{0^{\sigma}\}\;, & B_I&=\{0^{\alpha}\}\;, & B_R&=Y(z^*) + y^* \;.
\end{align*}

We first prove lower bounds on the sizes of $A$ and $B$, and later show that $L(a) = L(a+b)$ for all $a\in A$ and $b\in B$. By definitions of $A, B$, we see that $\abs{A} = 2^{\alpha}$ and $\abs{B} = \abs{Y(z^*)}\ge 2^{\rho - d}$.
From $\rho=\abs{R}\geq n - \abs{\Nbr(I) \cup I}\geq n - \alpha d \ell^2/n$, we conclude that $\abs{B} \ge 2^{n - d - \alpha d\ell^2/n}$, as desired.

It remains to show that for all $a\in A, b\in B: L(a) = L(a + b)$. To do that, by the decomposition of $L$, it suffices to prove the corresponding claims for $Q_1$ and $Q_2$.

We first claim that for all $(a_I, a_S)\in A_I\times A_S$ and $(b_I, b_S)\in B_I\times B_S$, it holds that $Q_1(a_I, a_S) = Q_1(a_I + b_I, a_S + b_S)$. Indeed, since $A_S = B_S = \{0^{\sigma}\}$ and $B_I = \{0^{\alpha}\}$, we have that $a_I = a_I + b_I$, and $a_S = a_S + b_S$; and the claim follows.

Next, we claim that for all $(a_R, a_S)\in A_R\times A_S$ and $(b_R, b_S)\in B_R\times B_S$, it holds that $Q_2(a_R, a_S) = Q_2(a_R + b_R, a_S + b_S)$. Since $A_S = B_S = \{0^{\sigma}\}$ and $A_R = \{y^*\}$, it suffices to show that $Q_2(y^*, 0^{\sigma}) = Q_2(y^* + b_R, 0^{\sigma})$.
As $y^*\in Y(z^*)$, we have that $Q_2(y^*, 0^{\sigma}) = z^*$.
Also, since $B_R = y^* + Y(z^*)$, $y^* + b_R\in Y(z^*)$ and hence, $Q_2(y^* + b_R, 0^{\sigma}) = z^* = Q_2(y^*, 0^{\sigma})$ as desired.

With these two claims, we have that for all $a\in A, b\in B, L(a) = L(a + b)$ as desired.
\end{proof}

Equipped with \cref{lemma: uniform decision tree key lemma}, we are ready to prove \cref{theorem: uniform decision tree theorem techinical}.
\TechnicalSingle*
\begin{proof}
By way of contradiction, assume there exists an oblivious $\ell$-local decision tree $T$ of depth $d$ computing $\Disp$.
Let $T$ be such that $T(x) = h(L(x))$ where $L: \zo^n\to \zo^d$ is an $\ell$-local map and $h: \zo^d \to \zo$.

We apply \cref{lemma: uniform decision tree key lemma} to $L$ with $\alpha=k+1$ to obtain $A, B \subseteq \zo^n$ such that $\abs{A} \ge 2^{\alpha}, \abs{B} \ge 2^{n - d - \alpha d \ell^2/n}$ and for all $a\in A, b\in B$, it holds that $L(a) = L(a + b)$.

For $z\in \zo$, let $A_z = \{y\in A: h(L(y)) = z\}$.
Without loss of generality, let $\abs{A_0} \ge \abs{A} / 2$.
Then, for all $a_0\in A_0$ and $b\in B$, it holds that $L(a_0) = L(a_0 + b)$.
Since $h(L(a_0)) = 0$, we infer that $h(L(a_0 + b)) = h(L(a_0)) = 0$.
Thus, $T(A_0 + B) = h(L(A_0 + B)) = \{0\}$.
However, as $\abs{A_0} \ge 2^{\alpha-1}\ge 2^k$ and $\abs{B}\ge 2^{n - d - \alpha d \ell^2/n}\ge 2^k$, we obtain a contradiction to the fact that $T$ is a sumset disperser for min-entropy $k$.
\end{proof}

\section{Multi-output oblivious local decision trees}\label{sec: noise sensitivity based lower bound}

In this section, we establish lower bounds for multi-output oblivious local decision trees by exploiting the fact that modifying a small number of input bits induces only a small change in the Hamming distance of the outputs of local maps. We use this observation to identify sumsets that any local map sends to a small number of outputs and, consequently, that are also mapped to a small number of outputs by any oblivious local decision tree. To prove \cref{thm:MainMultiple}, we then construct a function (specifically, a sumset condenser) for which such sumsets must be mapped to many distinct outputs.

In \cref{sec:condenser}, we will construct the following condensers for sumset sources, which will serve as the hard function for our lower bound.
\begin{restatable}{theorem}{thmCondenser}\label{thm: explicit sumset condenser}
There exist universal constants $C' > C > 1, \gamma > 0$ with the following property.
For every constant $\delta > 0$ and $n, m, k_1, k_2\in \N$ such that $k_1\ge k_2 \ge \log^C(n)$ and $m \le (1 - \delta)k_1$, there exists a polynomial-time computable function $\Cond: \zo^n \to \zo^m$ that is a $(k_1, k_2, k_{\out})$ sumset condenser  with error $n^{-\gamma}$, where $k_{\out} = m - \log^{C'}(n)$.
\end{restatable}

In \cref{subsec: technical multioutput lb}, we will prove the following technical version of \cref{thm:MainMultiple}.
\begin{restatable}[Technical version of \cref{thm:MainMultiple}]{theorem}{TechnicalMulti}\label{thm: uniform local trees cannot compute sumset condensers}
Let $n, m, \ell, \alpha, d$ be such that $2\ell d / n\le \alpha\le d$ and $m\le n$. 
Let $\Cond: \zo^n \to \zo^m$ be a $(k_1, k_2, k_{\out})$ sumset condenser with error $< \frac{1}{2}$, where 
\[
k_1 = n - d,\; k_2 = (\alpha n/2d\ell) \log(2d \ell / \alpha),\; k_{\out} = \alpha(\log(d/\alpha) + 3)\;.
\]
Then any oblivious $\ell$-local decision tree computing $\Cond$ requires depth at least $d+1$. 
\end{restatable}

We now prove \cref{thm:MainMultiple} using  the sumset condenser from \cref{thm: explicit sumset condenser} as the hard function in \cref{thm: uniform local trees cannot compute sumset condensers}.

\MainThmMultipleOutput*
\begin{proof}
Let $C'>C>1$  be the constants from \cref{thm: explicit sumset condenser}. In the following, we will assume that $\ell \leq O(n / \log^{C'+3}(n))$. Let $\Cond: \zo^n \to \zo^m$ be the $(k_1, k_2, k_{\out})$ sumset condenser from \cref{thm: explicit sumset condenser} for 
\[
k_1=n/\log(\log(n)),\;
k_2=\log^C(n),\;
k_\out=\ell\cdot\log^{C'+2}(n),\;
m=2k_\out\;.
\]
First we check that the parameters satisfy the requirements of \cref{thm: explicit sumset condenser}. Indeed, $k_1\geq k_2\geq \log^C(n)$, and $m\leq (1-\delta)k_1$ for every $\ell \leq O(n / \log^{C'+3}(n))$. Finally, $k_\out\leq m-\log^{C'}(n)$ holds from $k_\out\geq \log^{C'}(n)$.

We now apply \cref{thm: uniform local trees cannot compute sumset condensers} with $d=n-n/\log(\log(n))$ and $\alpha = \ell\cdot \log^C(n)$, and use $\Cond$ as the hard function. To see that we satisfy the requirements of \cref{thm: uniform local trees cannot compute sumset condensers}, we note that 
$2\ell d/n \leq \alpha\leq d$ holds from $C'>C$, and 
$m\leq n$ for all $\ell\leq O(n / \log^{C'+3}(n))$. Moreover, $k_1=n-d$,  $k_2< (\alpha n/2d\ell) \log(2d \ell / \alpha)$, and $k_\out>\alpha(\log(d/\alpha) + 3)$. 

Therefore, the $(k_1, k_2, k_\out)$ sumset condenser $\Cond$ requires oblivious $\ell$-local decision trees of depth at least $d=n-n/\log\log{n}$.
\end{proof}

\subsection{\texorpdfstring{Proof of \cref{thm: uniform local trees cannot compute sumset condensers}}{Proof of the lower bound}}\label{subsec: technical multioutput lb}

In this subsection, we prove \cref{thm: uniform local trees cannot compute sumset condensers}. To this end, we first prove the following lemma concerning the concentration of certain sumset sources when acted on by a local map.

\begin{lemma}[Sumset concentration lemma]\label{lemma: sumset concentration}
Let $n, d, \ell, \alpha\in \N$ be such that $2d \ell/n \le \alpha \le d$.
Then, for any $\ell$-local map $L: \zo^n \to \zo^d$, there exist sets $A, B\subseteq \zo^n$ with $\abs{A}\ge 2^{n-d}, \abs{B} \ge \binom{n}{\le (\alpha n / 2d \ell)}$ such that $\abs{L(A+B)} \le \binom{d}{\le \alpha}$.
\end{lemma}
\begin{proof}
By an averaging argument there exists $y_0\in \zo^d$ such that $\abs{L^{-1}(y_0)} \ge 2^{n-d}$. Let $A = L^{-1}(y_0)$.
For $x\in \zo^n$, define $\Nbr(x) = \{i\in [d]: L_i \textrm{ depends on some $j\in [n]$ s.t. } x_j = 1\}$.
Let $B = \{b\in \zo^n: \abs{\Nbr(b)}\le \alpha\}$.

We claim that $\abs{B}\ge \binom{n}{\le (\alpha n / 2\ell d)}$. 
To this end, consider a bipartite graph with the left vertex set $[n]$, the right vertex set $[d]$, and an edge $(j, i)$ in this bipartite graph if the $\ell$-local function $L_i$ depends on input $j$.
Since $L$ is an $\ell$-local map, the total number of edges in this graph is at most $d\ell$.
Hence, the average left degree in this graph is at most $d\ell/n$. 
For $w \in \N$, consider the quantity $\E_{x\in \zo^n, \abs{x} \le w}[\abs{\Nbr(x)}]$, obtained by choosing a random subset of size at most $w$ of the left vertices of $G$ and computing the total number of neighbors on the right. Since the average left degree in $G$ is at most $d \ell/n$, we infer that $\E_{x\in \zo^n, \abs{x} \le w}[\abs{\Nbr(x)}] \le wd\ell/n$.
We now apply Markov's inequality to obtain that $\Pr_{x\in \zo^n, \abs{x}\le w}[\abs{\Nbr(x)}\ge (4/3)wd\ell / n]\le \frac{3}{4}$.
Hence, $\Pr_{x\in \zo^n, \abs{x}\le w}[\abs{\Nbr(x)}\le (4/3)wd \ell / n]\ge \frac{1}{4}$.
Setting $w = \frac{3\alpha n}{4d\ell}$, we conclude that $\abs{B}\ge \frac{1}{4}\binom{n}{\le (3\alpha n/4d \ell)}\ge \binom{n}{\le (\alpha n / 2d \ell)}$, as desired.

Lastly, we show that $\abs{L(A+B)}\le \binom{d}{\le \alpha}$, proving our claim.
Let $a\in A, b\in B$ be arbitrary. 
Using the definition of $\Nbr$, we observe that $L(a+b)$ and $L(a)$ may only differ in output bits with indices in $\Nbr(b)$. Hence, the Hamming distance between $L(a+b)$ and $L(a)=y_0$ is at most $\abs{\Nbr(b)}$. By our choice of $B$, we have that $\abs{\Nbr(b)}\le \alpha$.
Since there are $\binom{d}{\le \alpha}$ many strings at Hamming distance at most $\alpha$ from $y_0$, we conclude that $\abs{L(A+B)}\le \binom{d}{\le \alpha}$, as desired.
\end{proof}

Equipped with \cref{lemma: sumset concentration}, we are ready to prove \cref{thm: uniform local trees cannot compute sumset condensers}.

\TechnicalMulti*
\begin{proof}
Assume that there exists an oblivious $\ell$-local decision tree $T$ of depth $d$ computing $\Cond$.
Let $T = h(L(\cdot))$ where $L: \zo^n\to \zo^d$ is an $\ell$-local map and $h$ is an arbitrary function.

We apply \cref{lemma: sumset concentration} to $L$ to infer that there exist $A, B\subseteq \zo^n$ such that $\abs{A}\ge 2^{n-d} = 2^{k_1}, \abs{B}\ge \binom{n}{\le \alpha n / 2d\ell}$, and $\abs{T(A+B)}\le \abs{L(A+B)}\le 2^{\alpha(\log(d/\alpha) + 2)}$. We note that $\abs{B} \ge 2^{k_2}$ since 
\[
\abs{B} \ge \binom{n}{\alpha n / 2d\ell} \ge \left(2d\ell / \alpha\right)^{\alpha n / 2d\ell} = 2^{k_2}\;.
\]

This implies that the output of $T$ on the sumset source $A+B$ is at statistical distance $\ge \frac{1}{2}$ from all distributions with min-entropy $\alpha(\log(d/\alpha) + 2) + 1 \le \alpha(\log(d/\alpha) + 3)=k_{\out}$. However, this contradicts the fact that $T$ computes $\Cond$. 
\end{proof}

\section{Construction of a sumset condenser}\label{sec:condenser}

In this section, we prove \cref{thm: explicit sumset condenser}, i.e., we construct a sumset condenser that serves as the hard function in \cref{thm:MainMultiple}. 
\thmCondenser*

\begin{proof}
Let $C_{\ref{thm: sumset extractor}}, \delta_{\ref{thm: sumset extractor}}, \gamma_{\ref{thm: sumset extractor}}$ be the universal constants from \cref{thm: sumset extractor}.
We let our universal constant $C > 0$ be large enough constant such that $C \ge C_{\ref{thm: sumset extractor}}$ and $C \ge 3 / \delta_{\ref{thm: sumset extractor}}$.

We let $\Ext: \zo^n \to \zo^{m_0}$ be the sumset extractor from \cref{thm: sumset extractor} for $(n, \log^C(n), \log^C(n))$ sumset sources, where $m_0=(\log^C(n))^{\delta_{\ref{thm: sumset extractor}}}\geq \log^3(n)$ by our choice of $C\geq 3 / \delta_{\ref{thm: sumset extractor}}$. Note that the choice $C \ge C_{\ref{thm: sumset extractor}}$ guarantees that we satisfy the entropy requirement of \cref{thm: sumset extractor}.

We let $\LExt: \zo^n \times \zo^d \to \zo^m$ be the linear seeded extractor from \cref{thm: trevisan's extractor} for min-entropy $k_1$ with error $\eps_0=n^{-\gamma_{\ref{thm: sumset extractor}}}$, where $d= O\left(\frac{\log^2(n/\eps_0)}{\log(k_1 / m)}\right)=O(\log^2(n))$. Note that our parameters $m\leq k_1\leq n$ satisfy the requirements of \cref{thm: trevisan's extractor}.

We are now ready to present our construction of the condenser $\Cond$.
On input $x\in\zo^n$, we compute $y = \Ext(x)\in\zo^{m_0}$. We then take the prefix $y_{\pre}$ of $y$ of length $d$, and output $\LExt(x, y_{\pre})\in\zo^m$. We note that for all large enough~$n$, the length of $y$ is $m\geq \log^3(n) \geq d$ as $d=O(\log^2(n))$.

We now analyze our construction. Let $\X = \X_1 + \X_2$ be an arbitrary $(n, k_1, k_2)$ sumset source. Without loss of generality we assume that $\X_2$ is an $(n,\log^C(n))$ flat source. Indeed, $\X_2$ is a convex combination of such flat sources (see, e.g., \cite[Lemma~6.10]{V12}). Then, $\X$ is a convex combination of sumset sources where the second source is a flat source with min-entropy exactly $\log^C(n)$. Now condensing from each such source, we obtain a condenser (with the exact same parameters) from the original source $\X$.

As $\LExt$ is a linear seeded extractor, our output distribution can be written as $\LExt(\X, \Y_{\pre}) = \LExt(\X_1, \Y_{\pre}) + \LExt(\X_2, \Y_{\pre})$.
We first argue that $\LExt(\X_1, \Y_{\pre})$ is close to uniform.
As $\Ext$ is a strong two source extractor, we are guaranteed that $(\X_1, \Y) \approx_{\eps_0} (\X_1, \U_{m_0})$ where $\eps_0 = n^{-\gamma_{\ref{thm: sumset extractor}}}$.
Since $\Y_{\pre}$ is a prefix of $\Y$, using the data processing inequality (\cref{lem: data processing inequality}), we obtain that $(\X_1, \Y_{\pre}) \approx_{\eps_0} (\X_1, \U_d)$. From this, using the triangle inequality and the data processing inequality again (\cref{lem: data processing inequality}), we conclude that $\LExt(\X_1, \Y_{\pre}) \approx_{2\eps_0} \U_m$ (because $\X_1$ has min-entropy $k_1$ required by $\LExt$).

We now show that adding $\LExt(\X_2, \Y_{\pre})$ to the nearly uniform distribution $\LExt(\X_1, \Y_{\pre})$ does not make it lose much of entropy. 
Towards this, we first observe that as $\X_2$ is a flat source with min-entropy exactly $\log^C(n)$, we have that $\abs{\supp(\X_2)} = 2^{\log^C(n)}$.
Using the fact that the domain of $\Y_{\pre}$ is $\zo^d$, we infer that $\abs{\supp(\LExt(\X_2, \Y_{\pre}))} \le 2^{d + \log^C(n)}$.

We apply the chain rule for smooth min-entropy (\cref{lem: chain rule smooth min entropy}) to random variables $\LExt(\X_1, \Y_{\pre})$ and $\LExt(\X_2, \Y_{\pre})$ with parameters $\eps_{\ref{lem: chain rule smooth min entropy}} = 2\eps_0$, $\delta_{\ref{lem: chain rule smooth min entropy}}=\eps_0$, and $\gamma_{\ref{lem: chain rule smooth min entropy}} = \sqrt{2\eps_0}$ to infer that for at least $1 - 6\sqrt{\eps_0}$ fraction of fixings $z_2 \sim \LExt(\X_2, \Y_{\pre})$, it holds that $\LExt(\X_1, \Y_{\pre}) | (\LExt(\X_1, \Y_{\pre}) = z_2)$ has statistical distance at most $\sqrt{2\eps_0}$ 
from a distribution with min-entropy at least $m - d - \log^C(n) - \log(1/\eps_0) = m - d - \log^C(n) - \gamma_{\ref{thm: sumset extractor}}\log(n)$. 
Conditioned on any such fixing, our output distribution becomes $\LExt(\X, \Y_{\pre}) = \LExt(\X_1, \Y_{\pre}) + z_2$ and so it also has statistical distance at most $\sqrt{2\eps_0}$ 
from a distribution with min-entropy at least $m - d - \log^C(n) - \gamma_{\ref{thm: sumset extractor}}\log(n)$.

As $d = O(\log^2(n))$ and $\gamma_{\ref{thm: sumset extractor}}$ is a universal constant, we let $C'$ be a large enough universal constant and $\gamma$ to be a small enough universal constant so that for at least $1 - n^{-\gamma}/2$ fraction of fixings $z_2 \sim \LExt(\X_2, \Y_{\pre})$, it holds that $\LExt(\X, \Y_{\pre}) | (\LExt(\X_1, \Y_{\pre}) = z_2)$ has statistical distance at most $n^{-\gamma}/2$ from a distribution with min-entropy at least $m - \log^{C'}(n)$. From this we conclude that $\LExt(\X, \Y_{\pre})$ has statistical distance at most $n^{-\gamma}$ from a distribution with min-entropy at least $m - \log^{C'}(n)$, as desired. 
\end{proof}

\section{Oblivious degree-2 decision trees}
\label{sec: oblivious deg 2 decision tree}
In this section, we prove maximal lower bounds against oblivious degree-$2$ decision trees, establishing \cref{thm:MainQuadratic}.

\begin{lemma}\label{lem: lb deg 2 decision tree}
Let $\Disp\colon\zo^n\to\zo$ be an $(n, k)$ directional affine disperser.
Then any oblivious degree-$2$ decision tree computing $\Disp$ requires depth at least $n - k + 1$.
\end{lemma}
\begin{proof}
Assume that there exists such a tree $T: \zo^n \to \zo$ of depth $d\le n-k$ computing~$\Disp$.  Let $T = h(Q(\cdot))$ where $Q\colon \zo^n \to \zo^d$ is a degree-$2$ map and $h: \zo^d \to \zo$ is an arbitrary function.
As $d < n$, there exist $u, v\in \zo^n$ such that $u\ne v$ and $Q(u) = Q(v)$.
Let $a = u + v$. Since $u\ne v$, we have that  $a\ne 0$. Let $L\colon \zo^n \to \zo^d$ be defined as
\[L(x) = Q(x) + Q(x + a)\;.\]

First we note that $L(u) = Q(u) + Q(v) = 0$.
As $Q$ is a degree $2$ map, $L$ is a degree $1$ map, i.e., an affine map.
Since $0$ lies in the image of the affine map $L$, there exists an affine subspace $U\subseteq \zo^n$ of dimension at least $n - d$ such that $L(U) = 0$.
Hence, for all $x\in U$, it holds that $Q(x) + Q(x + a) = 0$, implying $Q(x) = Q(x + a)$.
Therefore, for all $x\in U$, we obtain that $T(x) = h(Q(x)) = h(Q(x+a)) = T(x+a)$.
However, this implies that for all $x\in U$, $T(x) + T(x + a) = 0$.
As $U$ is a subspace of dimension at least $n - d \ge n - (n - k) = k$, and $a\ne 0$, we get a contradiction.
\end{proof}

We instantiate \cref{lem: lb deg 2 decision tree} with the construction of directional affine dispersers from~\cite{LZ24} (\cref{thm: explicit dir affine disperser}) to prove \cref{thm:MainQuadratic}.
\MainThmQuadratic*
\begin{proof}
We let $\Disp\colon\zo^n\to\zo$ be the $(n, k)$ directional affine disperser from \cref{thm: explicit dir affine disperser} for $k=O(\log(n))$, and apply \cref{lem: lb deg 2 decision tree} to it.
\end{proof}

\section{Conclusion and Open problems} \label{sec: conclusion and open problems}
We conclude by providing open problems for various models that we studied.

\paragraph{Oblivious decision trees.}
\begin{itemize}
    \item Find an explicit function requiring depth overhead of $\Omega(n)$ for oblivious local decision trees. Obtaining such a result for single-output functions would yield an improved lower bound for oblivious bounded-width branching programs.
    \item Prove strong lower bounds (larger than $0.01 n$, obtained from variety dispersers of \cite{LZ19})  for oblivious degree-$d$ decision trees for $d\ge 3$. It is unclear how to directly extend our lower bound for $d=2$ against directional affine dispersers to obtain even a non-trivial lower bound for $d=3$.
\end{itemize}

\paragraph{Non-oblivious decision trees.}
Perhaps the main direction for future work is to extend our lower bounds to \emph{non-oblivious} local decision trees, even with substantially weaker parameters.
\begin{itemize}
    \item The most pressing open problem is to construct an explicit function which requires depth overhead of $\omega(1)$ for non-oblivious local decision trees. We hope our techniques regarding finding sumset structure in local computation can be generalized to obtain such a bound.
    
    \item Extend the lower bound of \cref{thm:MainSingle} or \cref{thm:MainMultiple} (with $m=o(n)$ outputs) to the non-oblivious case. Even for constant $\ell=16$, a size lower bound of $2^{n-o(n)}$ would give us a lower bound of $3.9n-o(n)$ against unbounded-depth circuits---the first major improvement in decades.
    As a stepping stone towards this, it is necessary to obtain a depth lower bound of $n - o(n)$ for $\ell = 16$.
    
    \item Extend the lower bound of \cref{thm:MainSingle} or \cref{thm:MainMultiple} (with $m=o(d)$ outputs) to the non-oblivious setting. Even in the much weaker regime of $\ell=\omega(1)$ and $d=\Omega(n)$, this would imply a super-linear lower bound for series-parallel circuits. Extending this further to $\ell=n^{0.01}$ and $d=\omega(n/\log(\log(n)))$ would imply a super-linear lower bound for log-depth circuits. Either result would resolve a problem open for nearly 50 years~\cite{V77}.
    
    \item Extend the lower bound of \cref{thm:MainQuadratic} to the non-oblivious case. Dispersers for degree-2 variety sources with $o(n)$ dependence on min-entropy would imply a lower bound of $3.11n$ on the size of unbounded-depth circuits, improving the best known bound~\cite{LT22}. Such dispersers will also require depth $n - o(n)$ when computed by degree-2 decision trees. Hence, we view resolving this as a necessary step towards obtaining the improved circuit lower bound. Our maximal lower bounds for oblivious degree-2 decision trees resolve this problem in the oblivious setting.
    
    \item Show that any $\ell$-local decision tree computing Majority requires depth at least $(n/\ell)\cdot \log(\ell)$.
    It has been conjectured that $\Sigma_3^k$ circuits computing Majority (for $k \le \sqrt{n\log(n)}$) require size at least $2^{(n / k)\log(k)}$ \cite{HJP95}. Such a lower bound would yield the above depth lower bound for local decision trees. Hence, we view resolving this question as an important stepping stone towards resolving that conjecture.
    We note that Lecomte et al. showed that any \emph{oblivious} $\ell$-local decision tree computing Majority requires depth at least $(n / \ell)\cdot \log(\ell)$, resolving the question in the oblivious setting~\cite{LRT22}. 
    \item 
    More broadly, it would be interesting to discover more structural properties of local decision trees and find connections to other areas. For instance, the Fourier-analytic structure of degree-$1$ decision trees (parity decision trees) has found  many interesting connections to learning theory, communication complexity and the log-rank conjecture (see, e.g., \cite{KM93, MS24, GTW21}).
\end{itemize}

\section*{Acknowledgements}
We are thankful to Eshan Chattopadhyay for many fruitful discussions on this topic.

\section*{AI Usage Statement}
The article was written entirely by the authors, and all underlying ideas, arguments, and proofs are solely the authors’ own. ChatGPT Pro was used exclusively for proofreading the final version of the paper.

\newpage
\printbibliography

\end{document}